\documentclass[12pt]{article}

\usepackage[utf8]{inputenc}
\usepackage{amsmath,amsfonts, amsthm}
\usepackage[a4paper, total={16cm, 26cm}]{geometry}
\usepackage[small]{titlesec}
\usepackage{graphicx,color}
\usepackage{mathrsfs}
\usepackage{hyperref}

\newtheorem{proposition}{Proposition}[section]

\title{\large\bf Learning a latent pattern of heterogeneity in the \\ \bigskip innovation rates of a time series of counts}

\author{
\begin{tabular}[t]{c @{\extracolsep{1.5cm}} c @{\extracolsep{1.5cm}} c}
{\normalsize Helton Graziadei} & {\normalsize Hedibert F. Lopes} & {\normalsize Paulo C. Marques F.} \\
\textit{\small USP - S\~ao Paulo} & \textit{\small Insper - S\~ao Paulo} & \textit{\small Insper - S\~ao Paulo}
\end{tabular}
}

\date{July 2019}

\begin{document}

\maketitle

\begin{abstract}
We develop a Bayesian hierarchical semiparametric model for phenomena related to time series of counts. The main feature of the model is its capability to learn a latent pattern of heterogeneity in the distribution of the process innovation rates, which are softly clustered through time with the help of a Dirichlet process placed at the top of the model hierarchy. The probabilistic forecasting capabilities of the model are put to test in the analysis of crime data in Pittsburgh, with favorable results.
\end{abstract}

\section{Introduction}

Time series of counts are associated with a multiplicity of phenomena in fields as diverse as epidemiology, econometrics, finance, environmental studies, and public policy \cite{weiss}. In this paper, we develop a model for this kind of data, taking into account the possible existence of heterogeneities in the process distribution as it evolves through time.

In a general setting, we consider a Markovian process, for which the current count is modeled as some functional of the count at the previous epoch, plus a stochastic innovation, whose expectation may be specific to the current epoch. Modeling these unobservable innovation rates hierarchically in a suitable way, we can learn a latent pattern of heterogeneity in their distribution, and this information can be incorporated in the forecasting of future counts.

Our investigation implements this general idea in a particular setting built from two main components. Initially, we generalize the first-order integer autoregressive model (INAR(1) model hereafter), introduced in the seminal papers of McKenzie \cite{mckenzie} and Al-Osh and Alzaid \cite{alosh}, allowing for different values of the innovation rates at different times. Subsequently, this generalized model is extended hierarchically, with the help of a Dirichlet process \cite{ferguson}. This gives us a semiparametric model, which, due to the properties of the Dirichlet process, is capable of clustering the values of the innovation rates through time, based on the information contained in the observed counts, thereby allowing us to identify different innovation regimes in the time evolution of the process. Forecasting within this probabilistic model is made straightforwardly through the appropriate posterior predictive distributions.

The paper is organized as follows. In Section \ref{sec:ginar}, we generalize the original INAR(1) model, allowing for distinct innovation rates at different epochs of the process. In Section \ref{sec:daug}, this generalized INAR(1) model is data augmented, leading to a conditional specification of the model which enables the derivation of simple forms for the model parameters and latent variables full conditional distributions. The necessary Dirichlet process definitions and properties are briefly reviewed in Section \ref{sec:dirichlet}. In Section \ref{sec:dpinar}, we set up the DP-INAR(1) model introducing a Dirichlet process at the top of the hierarchy developed in Section \ref{sec:ginar}. After the forms of the prior distributions have been specified, we derive simple closed forms for the full conditional distributions of the model parameters and latent variables. We carefully consider the choice of prior parameters in Section \ref{sec:priors}. In Section \ref{sec:forecast}, we show how to use the DP-INAR(1) model to do Bayesian forecasting. In Section \ref{sec:pitt}, we put all the analytical results to work in the forecasting of crime data in Pittsburgh, US. In this application, the DP-INAR(1) model outperforms the original INAR(1) model in the majority of the patrol areas.

\section{Generalized INAR(1) model}\label{sec:ginar}

We begin by generalizing the original INAR(1) model of McKenzie \cite{mckenzie} and Al-Osh and Alzaid \cite{alosh} as follows.

Let $\{Y_t\}_{t\geq 1}$ be an integer-valued time series, and let the \textit{innovations} $\{Z_t\}_{t\geq 2}$, given positive parameters $\{\lambda_t\}_{t\geq 2}$, be a sequence of conditionally independent $\text{Poisson}(\lambda_t)$ random variables. Given a parameter $\alpha\in[0,1]$, let $\{B_i(t):i\geq 0, t\geq 2\}$ be a family of conditionally independent and identically distributed $\text{Bernoulli}(\alpha)$ random variables. Furthermore, given all the parameters, assume that the innovations $\{Z_t\}_{t\geq 2}$ and the family $\{B_i(t):i\geq 0, t\geq 2\}$ are conditionally independent. The generalized INAR(1) model is defined by the functional relation
$$
  Y_t = \alpha\circ Y_{t-1} + Z_t,
$$
\noindent for $t\geq 2$, in which $\circ$ denotes the binomial thinning operator, defined by $\alpha\circ Y_{t-1}=\sum_{i=1}^{Y_{t-1}} B_i(t)$, if $Y_{t-1}>0$, and $\alpha\circ Y_{t-1}=0$, if $Y_{t-1}=0$. In the homogeneous case,  when all the $\lambda_t$'s are assumed to be equal, we recover the original INAR(1) model.

This model can be interpreted as specifying a birth-and-death process, in which, at epoch $t$, the number of cases $Y_t$ is equal to the new cases $Z_t$ plus the cases that survived from the previous epoch; the role of the binomial thinning operator being to remove a random number of the $Y_{t-1}$ cases present at the previous epoch $t-1$.

Let $y=(y_1,\dots,y_T)$ denote the values of an observed time series. For simplicity, we assume that $Y_1=y_1$ with probability one. Since the process $\{Y_t\}_{t\geq 1}$ is Markovian, the joint distribution of $Y_1,\dots,Y_T$, given parameters $\alpha$ and $\lambda=(\lambda_2,\dots,\lambda_T)$, can be factored as
$$
\Pr\{Y_1=y_1,\dots,Y_T=y_T\mid\alpha,\lambda\} = \prod_{t=2}^T \Pr\{Y_t=y_t\mid Y_{t-1}=y_{t-1},\alpha,\lambda_t\}.
$$

Since, with probability one, $\alpha\circ Y_{t-1}\leq Y_{t-1}$ and $Z_t\geq 0$, by the law of total probability and  the definition of the generalized INAR(1) model we have that
\begingroup
\addtolength{\jot}{0.25cm}
\begin{align*}
\Pr\{Y_t = y_t &\mid Y_{t-1} = y_{t-1},\alpha,\lambda_t\} = \Pr\{\alpha \circ Y_{t-1} + Z_t = y_t \mid Y_{t-1} = y_{t-1}, \alpha,\lambda_t\} \\ 
  &=  \Pr\!\left\{\sum_{i = 1}^{Y_{t-1}} B_i(t) + Z_t = y_t \,\;\Bigg\vert\;\, Y_{t-1} = y_{t-1},\alpha,\lambda_t\right\}  \\ 
  &= \sum_{m_t = 0}^{\min\{y_t,\,y_{t-1}\}} \Pr\!\left\{\sum_{i = 1}^{y_{t-1}} B_i(t) = m_t, Z_t = y_t - m_t \,\;\Bigg\vert\;\, \alpha,\lambda_t\right\} \\ 
  &= \sum_{m_t = 0}^{\min\{y_t,\,y_{t-1}\}} \Pr\!\left\{\sum_{i = 1}^{y_{t-1}} B_i(t) = m_t \,\;\Bigg\vert\;\, \alpha\right\} \Pr\{Z_t = y_t - m_t \mid\lambda_t\}.
\end{align*}
\endgroup

Hence, the generalized INAR(1) model likelihood function is given by
$$
L_y(\alpha,\lambda) = \prod_{t=2}^T \sum_{m_t = 0}^{\min\{y_{t-1},\,y_t\}} \binom{y_{t-1}}{m_t} \alpha^{m_t} (1-\alpha)^{y_{t-1} - m_t} \left( \frac{e^{-\lambda_t}\lambda_t^{y_t - m_t}}{(y_t - m_t)!} \right)\!.
$$

In the next section, we show how the introduction of certain latent (unobservable) random variables allows us to specify the generalized INAR(1) model in terms of a set of conditional distributions. This alternative representation leads to a factorization of the model joint distribution which is the key element propelling our Monte Carlo simulations.

\section{Data augmentation}\label{sec:daug}

In the generalized INAR(1) model, suppose that, in addition to the values of the counts $Y_1,\dots,Y_T$, we could observe the values of the {\it maturations} $M_t=\alpha\circ Y_{t-1}$. The $M_t$'s would tell us the number of cases that matured (survived) from the previous epoch, breaking down $Y_t$ into two parcels: maturations plus innovations.

This is an example of data augmentation \cite{tanner,vandyk}, in which the introduction of the unobservable maturations, with suitable conditional distributions, factors the model into more manageable pieces. Within this data augmentation scheme, we postulate that
$$
  M_t\mid \alpha, Y_{t-1} = y_{t-1} \sim \text{Binomial}(y_{t-1},\alpha),
$$
and
$$
  \Pr\{Y_t=y_t\mid M_t=m_t,\lambda_t\} =
      \frac{e^{-\lambda_t}\lambda_t^{y_t - m_t}}{(y_t - m_t)!}\,\mathbb{I}_{\{m_t,m_{t+1},\,\dots\}}(y_t),
$$
in which $\mathbb{I}_A$ denotes the indicator function of the set $A$, defined by $\mathbb{I}_A(x)=1$, if $x\in A$, and $\mathbb{I}_A(x)=0$, if $x\notin A$.

% TODO: trocar para m_t \geq 0?
Using the law of total probability and the product rule, we have that
\begingroup
\addtolength{\jot}{0.25cm}
\begin{align*}
\Pr\{Y_t &= y_t \mid Y_{t-1} = y_{t-1},\alpha,\lambda_t\} = \sum_{m_t=0}^{y_{t-1}}\Pr\{Y_t = y_t, M_t = m_t\mid Y_{t-1} = y_{t-1},\alpha,\lambda_t\} \\
  &= \sum_{m_t=0}^{y_{t-1}}\Pr\{Y_t=y_t \mid M_t=m_t,\lambda_t\} \times \Pr\{M_t=m_t\mid Y_{t-1}=y_{t-1},\alpha\},
\end{align*}
\endgroup
in which, following the data augmentation scheme, we took advantage of the appropriate conditional independences.

Since
\begingroup
\addtolength{\jot}{0.25cm}
\begin{align*}
  \mathbb{I}_{\{m_t,m_{t+1},\dots\}}(y_t) \times \mathbb{I}_{\{0,1,\dots,y_{t-1}\}}(m_t) &= \mathbb{I}_{\{0,1,\dots,y_t\}}(m_t) \times \mathbb{I}_{\{0,1,\dots,y_{t-1}\}}(m_t) \\ &= \mathbb{I}_{\{0,1,\dots,\min\{y_t,y_{t-1}\}\}}(m_t),
\end{align*}
\endgroup
comparing the expression above for $\Pr\{Y_t = y_t \mid Y_{t-1} = y_{t-1},\alpha,\lambda_t\}$ with the results in the previous section, we come to the conclusion that this is a valid data augmentation scheme, since it induces the same generalized INAR(1) model likelihood function. 

In the next section, we recollect the main definitions and results related to the Dirichlet Process which are necessary to build-up our semiparametric hierarchical model. The data augmentation scheme developed above will come in handy in the derivation of the full conditional distributions of the complete model.

\section{The Dirichlet process}\label{sec:dirichlet}

Suppose that we represent our uncertainties about quantities assuming values in a sampling space $\mathscr{X}$, with sigma-field $\mathscr{B}$, by means of an underlying probability space $(\Omega,\mathscr{F},\Pr)$.

A mapping $\mathbb{G}:\mathscr{B}\times\Omega\to[0,1]$ is a random probability measure if $\mathbb{G}(\,\cdot\,,\omega)$ is a probability measure over $(\mathscr{X},\mathscr{B})$, for every $\omega\in\Omega$, and $\mathbb{G}(B)=\mathbb{G}(B,\cdot\,)$ is a random variable, for each $B\in\mathscr{B}$.

Ferguson \cite{ferguson} defined a random probability measure $\mathbb{G}$ descriptively as follows. Let $\beta$ be a finite nonzero measure over $(\mathscr{X},\mathscr{B})$ and postulate that for each $\mathscr{B}$-measurable partition $\{B_1,\dots,B_k\}$ of $\mathscr{X}$ the random vector $(\mathbb{G}(B_1),\dots,\mathbb{G}(B_k))$ has the ordinary Dirichlet distribution with parameters $(\beta(B_1),\dots,\beta(B_k))$. In this case, we say that $\mathbb{G}$ is a Dirichlet process with base measure $\beta$, and use the notation $\mathbb{G}\sim\text{DP}(\beta)$. Ferguson proved that $\mathbb{G}$ is a properly defined random process in the sense of Kolmogorov's consistency theorem.

Defining the concentration parameter $\tau=\beta(\mathscr{X})$, and the base probability measure $G_0$ by $G_0(B)=\beta(B)/\beta(\mathscr{X})$, it follows from the usual properties of the Dirichlet distribution that $\text{E}[\mathbb{G}(B)] = G_0(B)$ and $\text{Var}[\mathbb{G}(B)] = G_0(B) (1-G_0(B))/(\tau+1)$, for every $B\in\mathscr{B}$. Therefore, $\mathbb{G}$ is centered on $G_0$, and $\tau$ controls the concentration of $\mathbb{G}$ around $G_0$. In terms of the concentration parameter and the base probability measure, we write $\mathbb{G}\sim\text{DP}(\tau\,G_0)$.

Inference with the Dirichlet process is tractable. In particular, Ferguson proved that the Dirichlet process is closed under sampling: if $X_1,\dots,X_n$ are conditionally independent and identically distributed, given $\mathbb{G}\sim\text{DP}(\tau\,G_0)$, such that $\Pr\{X_i\in B\mid\mathbb{G}=G\} = G(B)$, for every $B$ in $\mathscr{B}$, then 
$$
  \mathbb{G}\mid X_1=x_1,\dots,X_n=x_n\sim\text{DP}\!\left((\tau+n)\left(\frac{\tau}{\tau+n}\,G_0+\frac{1}{\tau+n} \sum_{i=1}^n \mathbb{I}_B(x_i)\right)\right).
$$

Notice that, using the law of total expectation, we have
\begingroup
\addtolength{\jot}{0.25cm}
\begin{align*}
  \Pr\{X_{n+1}\in B\mid X_1,\dots,X_n\} &= \text{E}[\Pr\{X_{n+1}\in B\mid \mathbb{G}, X_1,\dots,X_n\} \mid X_1,\dots,X_n] \\
  &= \text{E}[\Pr\{X_{n+1}\in B\mid \mathbb{G}\} \mid X_1,\dots,X_n] \\
  &= \text{E}[\mathbb{G}(B)\mid X_1,\dots,X_n],
\end{align*}
\endgroup
almost surely, for every $B$ in $\mathscr{B}$, in which the second equality follows from the conditional independence of the $X_i$'s. Hence, the posterior predictive distribution is
$$
  \Pr\{X_{n+1}\in B\mid X_1=x_1,\dots,X_n=x_n\} = \frac{\tau}{\tau+n}\,G_0(B) + \frac{1}{\tau+n} \sum_{i=1}^n I_B(x_i).
$$

This expression of the posterior predictive distribution unleashes important features of the Dirichlet process, thereby showing how it can be used as a modeling tool. In particular, it defines a data generating process known as the P{\'o}lya-Blackwell-MacQueen urn \cite{blackwell}. If we imagine the sequential generation of the $X_i$'s, for $i=1,\dots,n$, we see that a value is generated anew from $G_0$ with probability proportional to $\tau$, or we repeat one the previously generated values with probability proportional to its multiplicity. This shows that, almost surely, realizations of a Dirichlet process $\mathbb{G}$ are discrete probability measures, maybe with denumerably infinite support, depending on the nature of $G_0$. Also, this data generating process associated with the P{\'o}lya-Blackwell-MacQueen urn implies that the $X_i$'s are clustered, which is the main feature of the Dirichlet process that we rely on to build our semiparametric model. Antoniak \cite{antoniak} derived the conditional distribution of the number of distinct $X_i$'s, that is, the number of clusters $K$, given the concentration parameter $\tau$, as
$$
  \Pr\{K=k\mid\tau\} =  S(n,k) \, \tau^k \, \frac{\Gamma(\tau)}{\Gamma(\tau + n)} \, \mathbb{I}_{\{1, 2, \dots, n\}}(k), 
$$
in which $S(n,k)$ denotes the unsigned Stirling number of the first kind. 

In the next section, we place a Dirichlet process at the top of the hierarchy of the generalized INAR(1) model, completing the specification of our semiparametric model, thereby being able to represent our uncertainty about the values of the unobservable innovation rates $\lambda_t$'s, given the information contained in the observed counts. In doing so, we benefit from the clustering properties of the Dirichlet process described above, identifying different regimes for the innovation rates as the process evolves through time.

\section{DP-INAR(1) model}\label{sec:dpinar}

The DP-INAR(1) model completes the generalized INAR(1) model defined in Section \ref{sec:ginar}, placing a Dirichlet process at the top of the hierarchy. Formally, we model the innovation rates $\lambda_2,\dots,\lambda_T$, given $\mathbb{G}\sim\text{DP}(\tau\,G_0)$, as conditionally independent and identically distributed, with $\Pr\{\lambda_t\in B\mid \mathbb{G}=G\} = G(B)$, for every Borel set $B$. The prior distributions for $\alpha$ and $\tau$ are $\text{Beta}(a_0^{(\alpha)},b_0^{(\alpha)})$ and $\text{Gamma}(a_0^{(\tau)}, b_0^{(\tau)})$, respectively. The base probability measure $G_0$ is a $\text{Gamma}(a_0^{(G_0)}, b_0^{(G_0)})$ distribution. In Section \ref{sec:priors}, we discuss in detail the choice of prior parameters.

Figure  \ref{fig:model} displays a graphical representation of the DP-INAR(1) model. In the graph, absence of an arrow connecting two random objects means that they are conditionally independent given their parents (see \cite{jordan} for a witful discussion of graphical models).

Our next step is to derive the full conditional distributions for all latent variables and model parameters. For convenience, we adopt a simplified notation in the following derivations, using the same letter $p$ to denote different probability functions or densities, with distinctions made clear from the context.

\begin{figure}[t!]
\centering
\includegraphics[width=9cm]{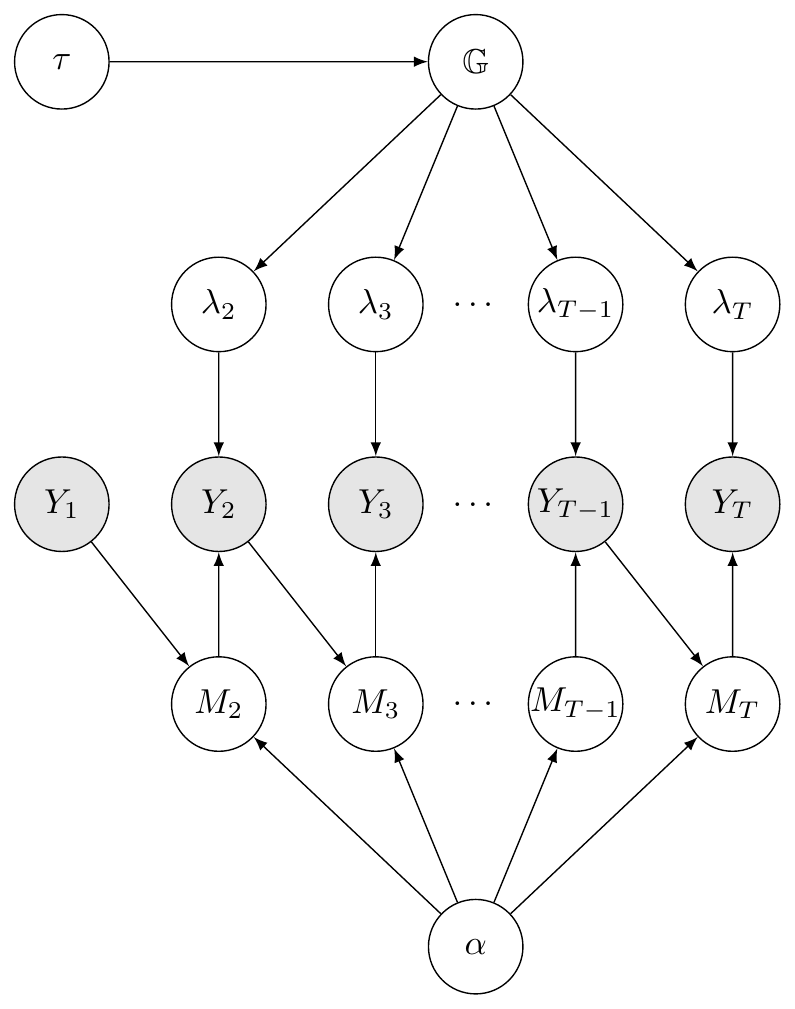}
\caption{The data augmented DP-INAR(1) model.}
\label{fig:model}
\end{figure}

Define $m=(m_2,\dots,m_T)$, and let $\mu_\mathbb{G}$ denote the distribution of $\mathbb{G}$. Marginalizing $\mathbb{G}$ on the graph, we have

\begingroup
\addtolength{\jot}{0.25cm}
\begin{align*} 
p(y,m,\alpha,\lambda) &= \int p(y,m,\alpha,\lambda \mid G)\,d\mu_{\mathbb{G}}(G) \\
&= \Bigg\{\prod_{t=2}^T p(y_t \mid m_t, \lambda_t) \ p(m_t \mid y_{t-1}, \alpha) \Bigg\} \times \pi(\alpha) \times \int \prod_{t=2}^T p(\lambda_t \mid G) \, d\mu_{\mathbb{G}}(G).
\end{align*} 
\endgroup

In this expression, the last integral is the joint distribution $p(\lambda_2,\dots,\lambda_T)$, pointing out that the random vector $(\lambda_2,\dots,\lambda_T)$ has an exchangeable distribution. Due to this distributional symmetry and the product rule, we can always make $p(\lambda_2,\dots,\lambda_T)$ depend on a certain $\lambda_t$ only through $p(\lambda_t\mid\lambda_{\setminus t})$, in which $\lambda_{\setminus t}$ denotes the vector $\lambda$ with the component $\lambda_t$ removed. Using the symbol $\boldsymbol{\propto}$ to denote proportionality up to a suitable normalization factor, and the label ``all others'' to designate the observed counts $y$, and all the other latent variables and model parameters, with the exception of the one under consideration, we have that
$$
p(\lambda_t \mid \text{all others}) \boldsymbol{\propto} p(y, m, \alpha, \lambda) \boldsymbol{\propto} p(\lambda_t \mid \lambda_{\setminus t}) \ p(y_t \mid m_t, \lambda_t) \boldsymbol{\propto} e^{-\lambda_t} \lambda_t^{y_t - m_t} \ p(\lambda_t \mid \lambda_{\setminus t}).
$$

Therefore, the P{\'o}lya-Blackwell-MacQueen urn process yields the full conditional distribution of $\lambda_t$ as the mixture
\begin{align*}
\displaystyle  \lambda_t \mid \text{all others} &\sim  \frac{\tau \cdot (b_0^{(G_0)})^{a_0^{(G_0)}} \cdot \Gamma(y_t - m_t + a_0^{(G_0)})}{\Gamma(a_0^{(G_0)}) (b_0^{(G_0)} + 1)^{y_t - m_t + a_0^{(G_0)}}} \times \text{Gamma}(y_t - m_t + a_0^{(G_0)},b_0^{(G_0)} + 1) \\ 
&\qquad+ \sum_{r \neq t} \lambda_r^{y_t - m_t} e^{-\lambda_r} \delta_{\{\lambda_r\}},
\end{align*}
in which $\delta_{\{\lambda_r\}}$ denotes a point mass at $\lambda_r$. In the former expression we suppressed the normalization constant which makes all mixture weights add up to one.

The derivations of the full conditionals for $\alpha$ and $m_t$ are straightforward.
$$
  \alpha \mid \text{all others} \sim \text{Beta}\!\left(a^{(\alpha)}_0 + \sum_{t=2}^T m_t, b^{(\alpha)}_0 + \sum_{t=2}^T (y_{t-1} - m_t)\right).
$$

$$ 
  p(m_t \mid \text{all others}) \boldsymbol{\propto} \frac{1}{m_t!(y_t - m_t)!(y_{t-1} - m_t)!} \Bigg(\frac{\alpha}{\lambda_t (1 - \alpha)} \Bigg)^{m_t} \mathbb{I}_{ \{0, 1, \ldots, \min{\{y_{t-1}, y_t\}} \}}(m_t).
$$

West \cite{west} shows how to derive the full conditional distribution of the concentration parameter $\tau$ in simple closed form, after the introduction of an auxiliary random variable $U$. Using this technique, we have the full conditionals
$$
  U \mid \text{all others} \sim \text{Beta}(\tau + 1, T - 1); 
$$

\begin{align*}
\displaystyle \tau \mid \text{all others} &\sim \frac{\Gamma(a_0^{(\tau)} + k )}{(b_0^{(\tau)} - \log u)^{a_0^{(\tau)}+k-1}} \times \text{Gamma}(a_0^{(\tau)} + k, b_0^{(\tau)} - \log u) \\ 
&\qquad+ \frac{(T-1) \cdot \Gamma(a_0^{(\tau)} + k - 1)}{(b_0^{(\tau)} - \log u)^{a_0^{(\tau)} + k - 1}} \times \text{Gamma}(a_0^{(\tau)} + k - 1, b_0^{(\tau)} - \log u), 
\end{align*}
in which we suppressed the normalization constant which makes the two mixture weights add up to one.

These full conditional distributions allow us to explore the model posterior distribution by coding a plain Gibbs sampler \cite{gamerman}. Experimentation with this Gibbs sampler shows that, as pointed out by Escobar and West \cite{escobar} in a similar context, we can improve mixing by resampling simultaneously the values of all $\lambda_t$'s inside the same cluster at the end of each iteration. Formally, let $(\lambda^*_1,\dots,\lambda^*_k)$ be the $k$ unique values among $(\lambda_2,\dots,\lambda_T)$ and define the number of occupants of cluster $j$ by $n_j = \sum_{t=2}^T \mathbb{I}_{\{\lambda^*_j\}}(\lambda_t)$. It follows that
$$
  \lambda_j^{*} \mid \text{all others} \sim \text{Gamma}\!\left(a_0^{(G_0)} + \sum_{t = 2}^T (y_t - m_t)\cdot \mathbb{I}_{\{\lambda^*_j\}}(\lambda_t), b_0^{(G_0)} + n_j\right).
$$
for $j = 1,\ldots, k$. After the $\lambda^*_j$'s are sampled from this distribution, we update the values of all $\lambda_t$'s inside each cluster by the corresponding $\lambda^*_j$.

In the next section, we discuss how to choose the prior parameters for the DP-INAR(1) model.

\section{Choice of prior parameters}\label{sec:priors}

Extending the original scheme proposed by Dorazio \cite{dorazio}, we choose the parameters $a_0^{(\tau)}$ and $b_0^{(\tau)}$ of the $\tau$ prior by minimizing the Kullback-Leibler divergence between the prior distribution of the number of clusters $K$ and a uniform discrete distribution on a suitable range. Using the results in Section \ref{sec:dirichlet}, the marginal probability function of $K$ can be computed as
$$
  \pi(k) = \int_{0}^{\infty} \Pr\{K=k\mid \tau\} \, \pi(\tau) \, d\tau = \frac{b_0^{(\tau)} S(T-1, k)}{\Gamma(a_0^{(\tau)})} I(a_0^{(\tau)}, b_0^{(\tau)}; k), 
$$
for $k=1,\dots,T-1$, in which
$$
  I(a_0^{(\tau)}, b_0^{(\tau)}; k) = \int_{0}^{\infty} \frac{\tau^{k + a_0^{(\tau)} - 1} \ e^{-b_0^{(\tau)} \tau} \ \Gamma(\tau)}{\Gamma(\tau+T-1)} \, d\tau.
$$
Using the information available about the phenomena under consideration to make a sensible choice for the integers $k_\text{min}$ and $k_\text{max}$, and letting $q$ be the probability function of a uniform discrete distribution on $\{k_\text{min},\dots,k_\text{max}\}$, that is
$$
  q(k) = \frac{1}{(k_\text{max}-k_\text{min}+1)} \, \mathbb{I}_{\{k_\text{min},\dots,k_\text{max}\}}(k),
$$
we find, by numerical integration and optimization, the values of $a_0^{(\tau)}$ and $b_0^{(\tau)}$ that minimize the Kullback-Leibler divergence
\begin{align*}
  \text{KL}[\pi &\;\Vert\; q] = \sum_{k=k_\text{min}}^{k_\text{max}} q(k) \log\left(\frac{q(k)}{\pi(k)}\right) \\
  &= {\footnotesize \text{(constant)}} + \log\Gamma(a_0^{(\tau)}) - a_0^{(\tau)}\log b_0^{(\tau)} - \frac{1}{(k_\text{max}-k_\text{min}+1)} \sum_{k=k_\text{min}}^{k_\text{max}} \log I(a_0^{(\tau)}, b_0^{(\tau)}; k).
\end{align*}

We choose the parameters $a_0^{(G_0)}$ and $b_0^{(G_0)}$ of the base probability density $g_0$ in a similar fashion, minimizing the Kullback-Leibler divergence between $g_0$ and a uniform distribution on a suitable range $[0,\lambda_\text{max}]$, in which $\lambda_\text{max}$ is chosen by taking into consideration the available information on the studied phenomena. Letting $h$ be a uniform density on $[0,\lambda_\text{max}]$, that is
$$
  h(\lambda) = \left(\frac{1}{\lambda_\text{max}}\right) \, \mathbb{I}_{[0,\lambda_\text{max}]}(\lambda),
$$
we find, by numerical optimization, the values of $a_0^{(G_0)}$ and $b_0^{(G_0)}$ that minimize the Kullback-Leibler divergence
\begin{align*}
 &\text{KL}[g_0 \;\Vert\; h] = \int_0^{\lambda_\text{max}} \left(\frac{1}{\lambda_\text{max}}\right) \log\left(\frac{1/\lambda_\text{max}}{g_0(\lambda)}\right) \, d\lambda \\
  &\quad= -\log \lambda_\text{max} - a_0^{(G_0)} \log b_0^{(G_0)} + \log \Gamma(a_0^{(G_0)})  - (a_0^{(G_0)} - 1)(\log \lambda_\text{max} - 1) + \frac{b_0^{(G_0)}\lambda_\text{max}}{2}.
\end{align*}

Choosing the parameters for the $\alpha$ prior is more straightforward, with $a_0^{(\alpha)}=b_0^{(\alpha)}=1$ being a natural choice.

\section{Bayesian forecasting}\label{sec:forecast}

The Gibbs sampler described in Section \ref{sec:dpinar} yields, marginally, a sample $\{\alpha^{(n)},\lambda^{(n)}\}_{n=1}^N$ from the posterior distribution. Note that, for $n=1,\dots,N$, we can obtain the number of clusters $k^{(n)}$ as the number of distinct entries in the vector $\lambda^{(n)}=(\lambda^{(n)}_1,\dots,\lambda^{(n)}_T)$. Uncertainty about future counts is represented by the $h$-steps-ahead posterior predictive distribution
$$
  Y_{T+h} \mid Y_1=y_1,\dots,Y_T=y_T,
$$
for some target $h\geq 1$. In particular, a pointwise forecast is obtained as a suitable summary of this posterior predictive distribution.

Using the law of total probability, the product rule, and simplifying the conditional independences in the model, we can write the posterior predictive probability function as
\begin{align*}
p(y_{T+h}\mid y_1,\dots,y_T) &= \int p(y_{T+h}\mid y_T,\alpha,\lambda_{T+1},\dots,\lambda_{T+h}) \\
  &\qquad\times \prod_{i=1}^h p(\lambda_{T+i}\mid\lambda_2,\dots,\lambda_{T+i-1}) \\
  &\qquad\times p(\alpha,\lambda_2,\dots,\lambda_T\mid y_1,\dots,y_T) \,d\alpha\,d\lambda_2\dots d\lambda_{T+h}.
\end{align*}

A nice property of the DP-INAR(1) model is that we can derive a simple analytical expression for the first factor in the integrand above.

\begin{proposition}\label{prop:pred}
The probability function of $Y_{t+h}$, given $Y_t=y_t$, $\alpha$, and $(\lambda_{t+1},\dots,\lambda_{t+h})$, can be writen as the convolution of a $\text{Bin}(y_t,\alpha^h)$ distribution and a $\text{Poisson}(\mu_h$) distribution,
$$
  p(y_{t+h}\mid y_t,\alpha,\lambda_{t+1},\dots,\lambda_{t+h}) = \sum_{m=0}^{\min\{y_t,y_{t+h}\}} \binom{y_t}{m} (\alpha^h)^m (1-\alpha^h)^{ y_t-m} \left(\frac{\mu_h^{y_{t+h}-m}e^{-\mu_h}}{(y_{t+h}-m)!}\right),
$$
in which
$$
  \mu_h = \sum_{i=1}^h \alpha^{h-i}\lambda_{t+i}.
$$
\end{proposition}

\begin{proof} 
We prove the result by induction. For $h=1$, using a simplified notation, the conditional moment generating function is given by
$$
  M_{Y_{t+1}\mid Y_t}(s) = \mathrm{E}\!\left[ e^{s Y_{t+1}}\mid Y_t\right] =
  \mathrm{E}\!\left[ e^{s(\alpha\circ Y_t+Z_{t+1})}\mid Y_t\right] =
  \mathrm{E}\!\left[ e^{s(\sum_{i=1}^{Y_t} B_i(t)+Z_{t+1})}\mid Y_t\right],
$$
But since $\{Z_t\}_{t\geq 2}$ is a sequence of conditionally independent random variables, which is also conditionally independent of $\{B_i(t):i\geq 0, t\geq 2\}$, we have that
$$ M_{Y_{t+1}\mid Y_t}(s) = \mathrm{E}\!\left[ e^{s\sum_{i=1}^{Y_t} B_i(t)}\mid Y_t\right] \mathrm{E}\!\left[ e^{sZ_{t+1}}\right] = (\alpha e^s+(1-\alpha))^{Y_t} \exp(\lambda_{t+1}(e^s-1)),$$ 
which is the product of the generating functions of a $\text{Binomial}(Y_t, \alpha)$ random variable and a $\text{Poisson}(\lambda_{t+1})$ random variable. Now, suppose the result holds for an arbitrary $h\geq 2$. Then,
\begingroup
\addtolength{\jot}{0.25cm}
\begin{align*}
  M_{Y_{t+h+1}\mid Y_t}(s) &= \mathrm{E}\!\left[e^{sY_{t+h+1}}\mid Y_t\right] = \mathrm{E}\!\left[\mathrm{E}\!\left[e^{sY_{t+h+1}}\mid Y_{t+h}\right]\mid Y_t\right] \\
  %&= \mathrm{E}\!\left[ (\alpha e^s+(1-\alpha))^{Y_{t+h}} \exp(\lambda_{t+h+1}(e^s-1))\mid Y_t \right] \\
  &= \mathrm{E}\!\left[e^{u Y_{t+h}} \mid Y_t \right]\exp(\lambda_{t+h+1}(e^s-1)),
\end{align*}
\endgroup
in which we defined $e^u=\alpha e^s+(1-\alpha)$. Consequently, from the induction hypothesis, we have that
\begin{align*}
  M_{Y_{t+h+1}\mid Y_t}(s) &= (\alpha^h e^u+(1-\alpha^h))^{Y_t} \exp(\mu_h(e^u-1))\exp(\lambda_{t+h+1}(e^s-1)) \\
  &= (\alpha^h (\alpha e^s+(1-\alpha))+(1-\alpha^h))^{Y_t} \exp(\mu_h((\alpha e^s+(1-\alpha))-1)) \\ &\quad\times \exp(\lambda_{t+h+1}(e^s-1)) \\
  &= (\alpha^{h+1} e^s+(1-\alpha^{h+1}))^{Y_t} \exp(\mu_{h+1}(e^s-1)), 
\end{align*}
in which $\mu_{h+1} = \alpha \mu_h + \lambda_{t+h+1}$. Hence, the result holds for $h+1$, completing the proof. 
\end{proof} 

% If the generating function can be factored as two generating functions which we can ``identify'', then the random variable is the sum of the two correspondent (independent) random variables, and we can compute the desired probabilities through a convolution.

Using the P{\'o}lya-Blackwell-MacQueen urn process repeatedly, for $n=1\dots,N$, we draw a sample $\{\lambda_{T+1}^{(n)},\dots,\lambda_{T+h}^{(n)}\}_{n=1}^N$ from $\prod_{i=1}^h p(\lambda_{T+i}\mid\lambda_2,\dots,\lambda_{T+i-1})$ sequentially as follows:
\begin{align*}
\lambda^{(n)}_{T+1} &\sim \frac{\tau}{\tau +T} \, G_0 + \frac{1}{\tau +T} \sum_{t=2}^{T} \delta_{\{\lambda^{(n)}_t\}}; \\
\lambda^{(n)}_{T+2} &\sim \frac{\tau}{\tau +T+1} \, G_0 + \frac{1}{\tau +T+1} \sum_{t=2}^{T+1} \delta_{\{\lambda^{(n)}_t\}}; \\
&\,\,\,\vdots \\
\lambda^{(n)}_{T+h} &\sim \frac{\tau}{\tau+T+h-1} \, G_0 + \frac{1}{\tau+T+h-1} \sum_{t=2}^{T+h-1} \delta_{\{\lambda^{(n)}_t\}}.
\end{align*}

Combining all these elements, we approximate the integral representation of the $h$-steps-ahead posterior predictive probability function by the Monte Carlo average
$$
p(y_{T+h}\mid y_1,\dots,y_T) \approx \frac{1}{N} \sum_{n=1}^N p(y_{T+h}\mid y_T,\alpha^{(n)},\lambda_{T+1}^{(n)},\dots,\lambda_{T+h}^{(n)}),
$$
for $y_{T+h}\geq 0$.

As a pointwise forecast $\hat{y}_{T+h}$, we compute the generalized median of the $h$-steps-ahead posterior predictive distribution, defined by
$$
  \hat{y}_{T+h} = \arg \min_{y_{T+h}\geq 0} \left|0.5 - \sum_{r=0}^{y_{T+h}} p(r\mid y_1,\dots,y_T)\right|.
$$

We use a form of cross-validation to evaluate the forecasting performance of the model. For an observed time series $y_1,\dots,y_T$, we pick some $T^*<T$, and treat the counts $y_{T^*},\dots,y_T$ as a holdout (test) sample. For $t\geq T^*$, we train the model conditioning only on the values $y_1,\dots,y_{t-1}$ and making an $h$-steps-ahead prediction $\hat{y}_{t+h}$. To score the forecast performance, we average the median deviations $|\hat{y}_{t+h}-y_{t+h}|$ over all predictions. This cross-validation procedure is depicted in Figure \ref{fig:cv}.

\begin{figure}[t!]
\centering
\includegraphics[width=13cm]{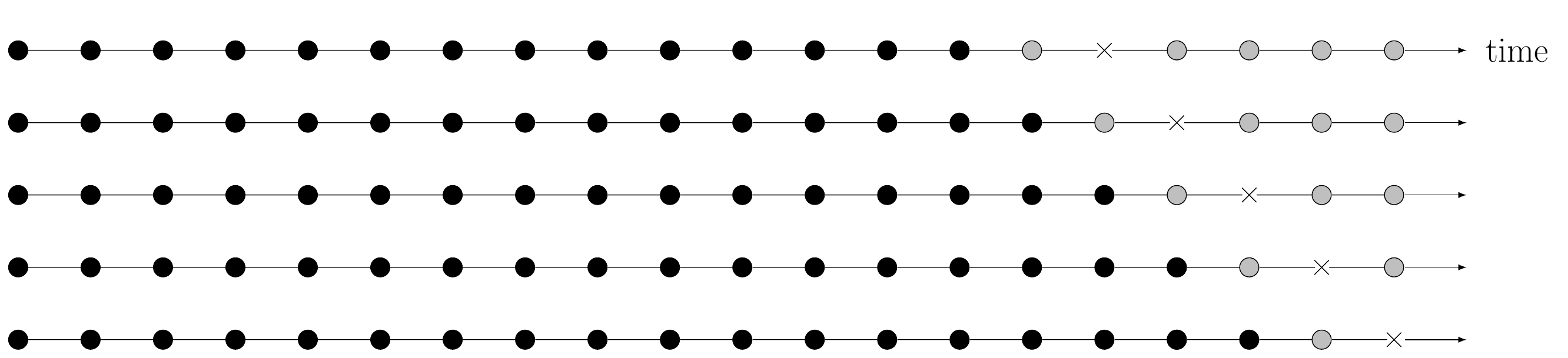}
\caption{Cross-validation scheme for two-steps-ahead predictions. For each line, the black dots indicate the training set. Predictions are made for the target epoch marked with an $\times$.}
\label{fig:cv}
\end{figure}

In the next section, we assess the forecasting performance of the DP-INAR(1) model, analyzing monthly time series of burglary occurrences in Pittsburgh, USA.

\section{Pittsburgh crime data}\label{sec:pitt}

In this section, we analyze monthly time series of burglary events in Pittsburgh, USA, from January $1990$ to December $2001$ \cite{crimedata}. In this dataset, each time series has a length of $144$ months and corresponds to a certain patrol area.

Figure \ref{fig:ts58} presents the time series for patrol area $58$, which displays substantial time heterogeneity and variation in the monthly counts of burglary events. In what follows, we use this patrol area $58$ to exemplify the model training procedure.

\begin{figure}[t!]
\centering
\includegraphics[width=16cm]{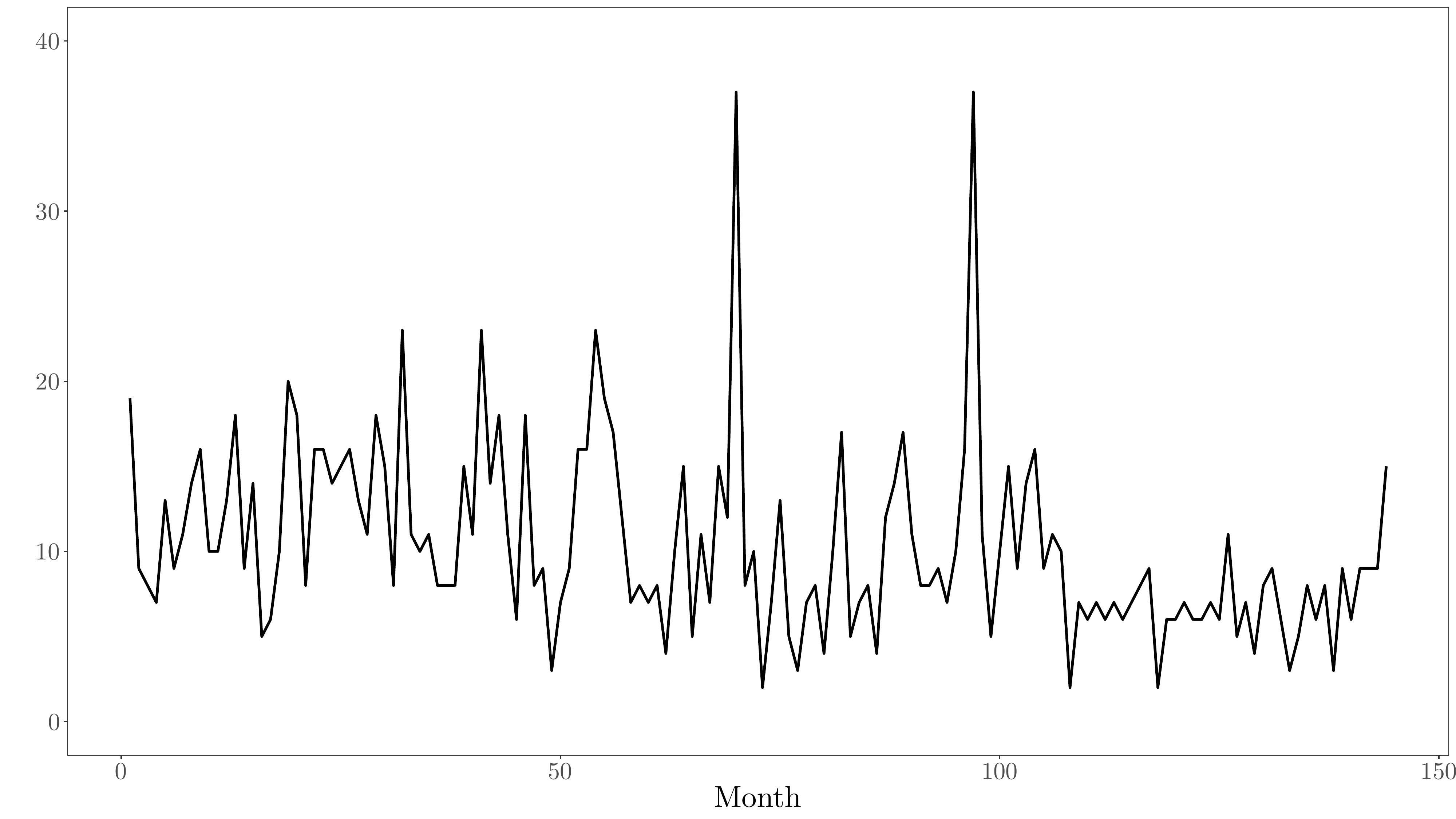}
\caption{Monthly burglary events for patrol area $58$.}
\label{fig:ts58}
\end{figure}

To determine the hyperparameters $a_0^{(\tau)}$ and $b_0^{(\tau)}$, the optimization procedure described in Section \ref{sec:priors}, with $k_{\min} = 1$ and $k_{\max} = 143$, yields $a_0^{(\tau)}=0.519$ and  $b_0^{(\tau)}=0.003$.

Using the procedure discussed in Section \ref{sec:priors}, we control the support of $G_0$ by choosing the value of $\lambda_{\max}$ to be the maximum observed count. Figure \ref{fig:optbase58} displays the level curves of $\text{KL}[g_0 \;\Vert\; h]$. The minimum is attained at $a_0^{(G_0)}=1.778$ and $b_0^{(G_0)}=0.096$.

For the thinning parameter $\alpha$, we adopt a uniform prior, choosing $a_0^{(\alpha)} = b_0^{(\alpha)} = 1$.

The marginal posterior distributions of parameters $\alpha$, $\lambda_{3}$, $\lambda_{18}$, and $\lambda_{96}$ are displayed in Figure \ref{fig:post58}. The posterior distribution of the thinning parameter $\alpha$ is reasonably concentrated, with posterior mean $0.19$, showing that the autoregressive component is not negligible. The posterior distributions of $\lambda_{3}$, $\lambda_{18}$ and $\lambda_{96}$ are fairly concentrated as well, with posterior means equal to $6.50$, $13.61$ and $32.01$, respectively, showing that different regimes of innovation rates were captured in the learning process. The Markov chains in Figure \ref{fig:chains58} indicate that proper mixing is achieved by the Gibbs sampler.

\begin{figure}[t!]
\centering
\includegraphics[width=16cm]{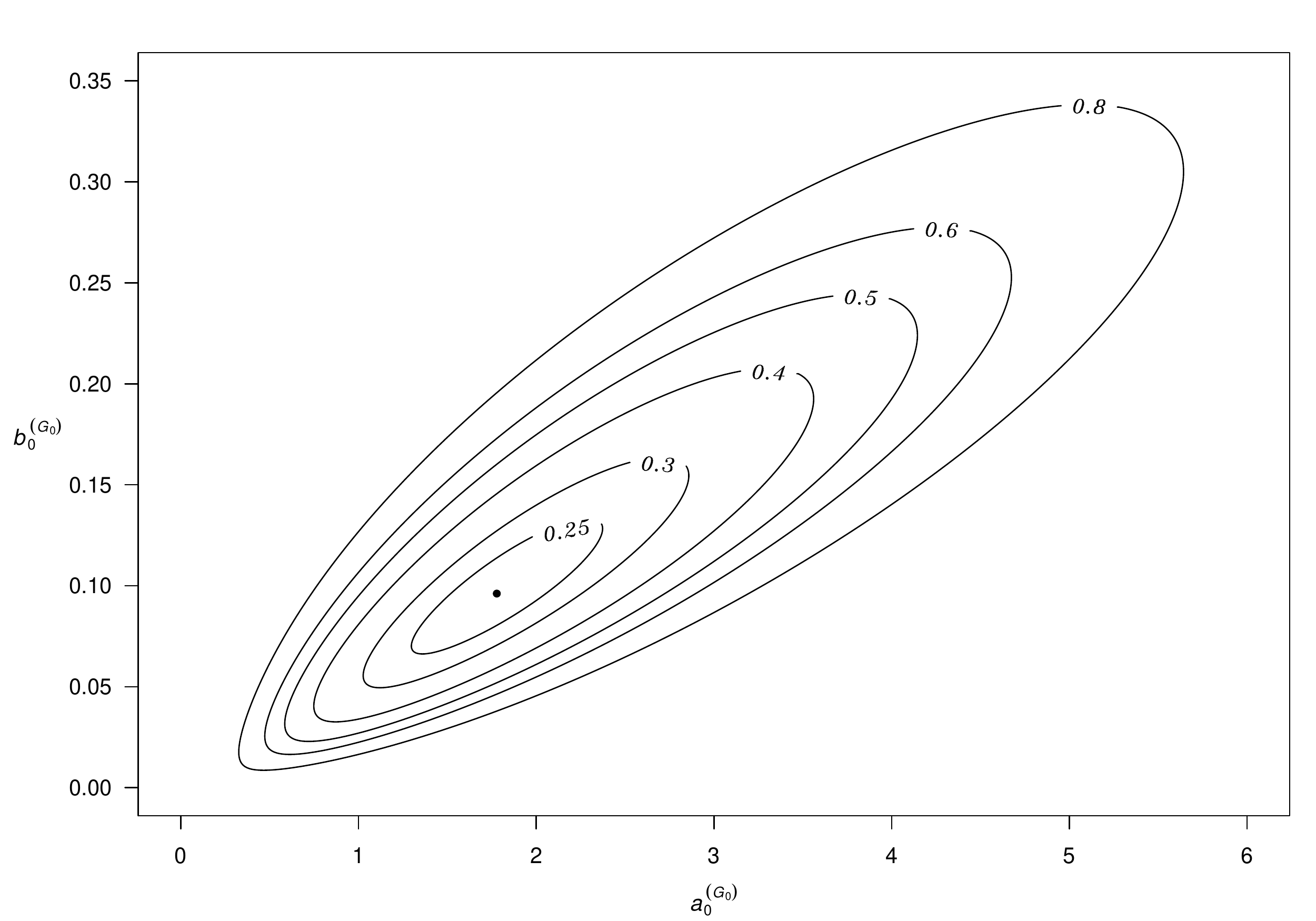}
\caption{Level curves of the Kullback-Leibler divergence associated with the optimization of the base measure hyperparameters for patrol area 58.}
\label{fig:optbase58}
\end{figure}

\begin{figure}[t!]
\centering
\includegraphics[width=16cm, height=14cm]{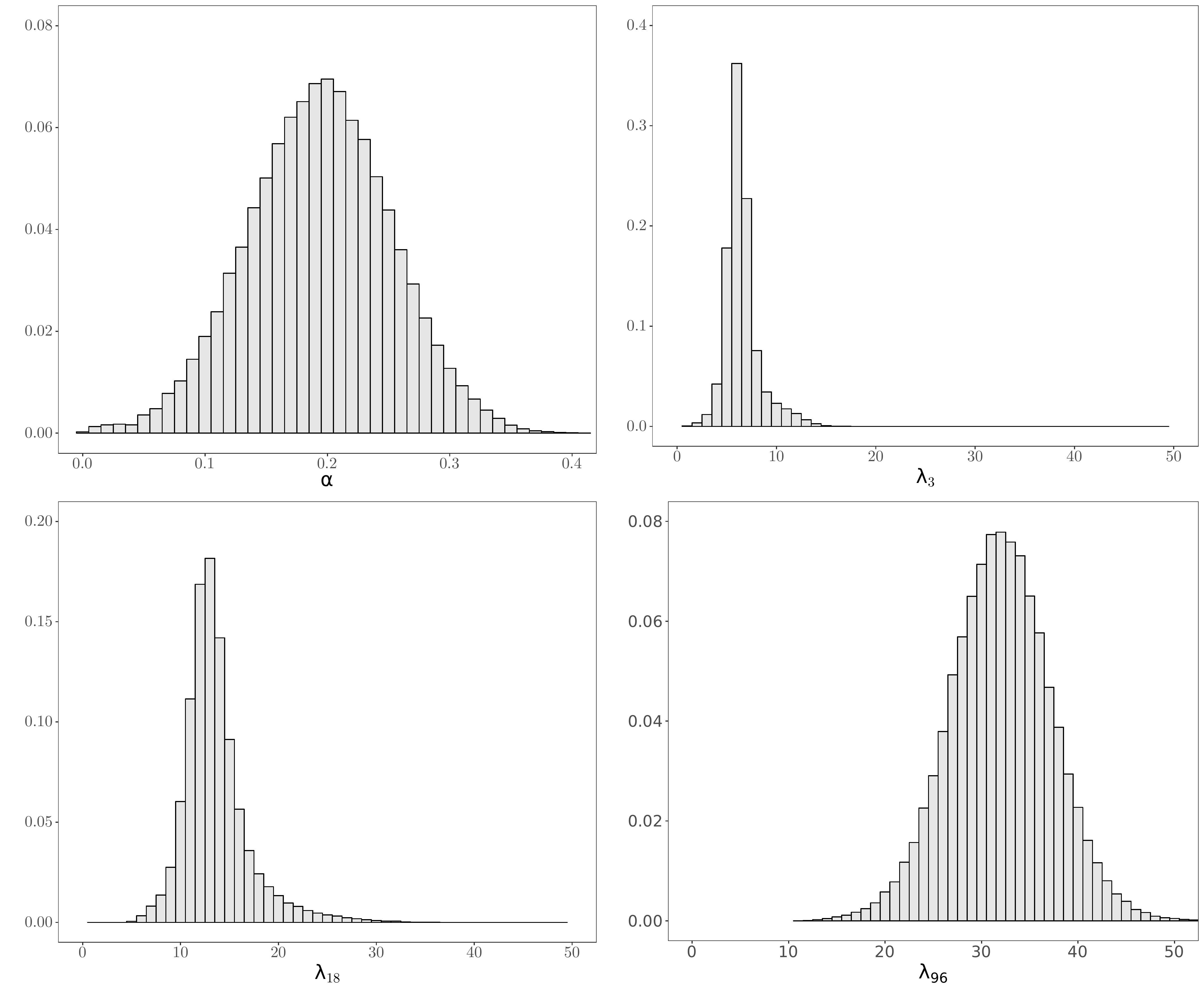}
\caption{Marginal posterior distributions of parameters $\alpha$, $\lambda_{3}$, $\lambda_{18}$, and $\lambda_{96}$, for patrol area 58.} 
\label{fig:post58}
\end{figure}

\begin{figure}[t!]
\centering
\includegraphics[width=16cm, height=14cm]{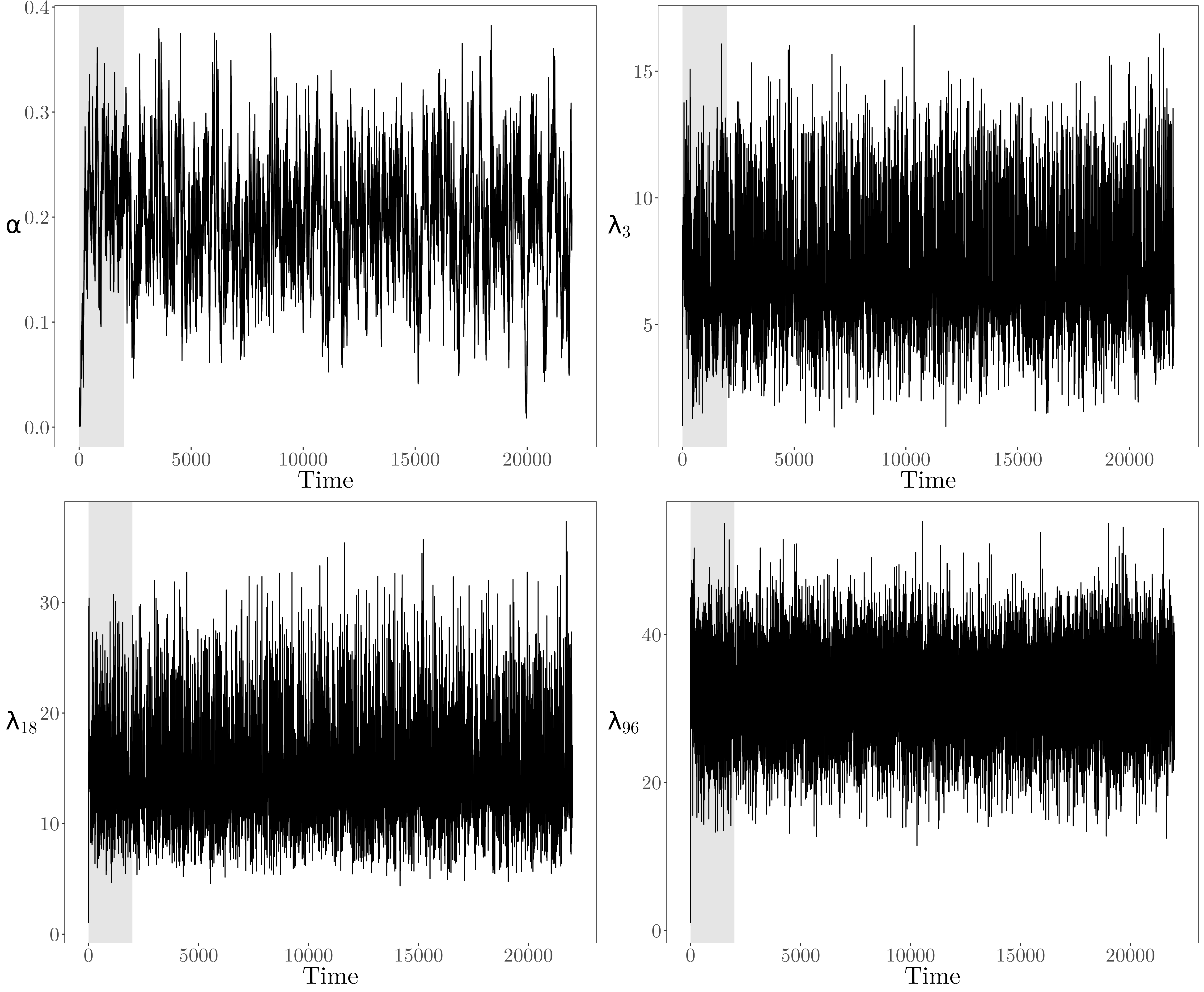}
\caption{Markov chains associated with the marginal posterior distributions of parameters $\alpha$, $\lambda_{3}$, $\lambda_{18}$, and $\lambda_{96}$, for patrol area 58. The gray rectangles indicate the burn-in periods.}
\label{fig:chains58}
\end{figure}

Figure \ref{fig:postk58} shows both the prior and posterior distributions of the number of clusters $K$. While the prior distribution is reasonably flat in the range $1$ to $143$, the posterior distribution is concentrated around $7$, the posterior mode.

\begin{figure}[t!]
\centering
\includegraphics[width=16cm, height=12cm]{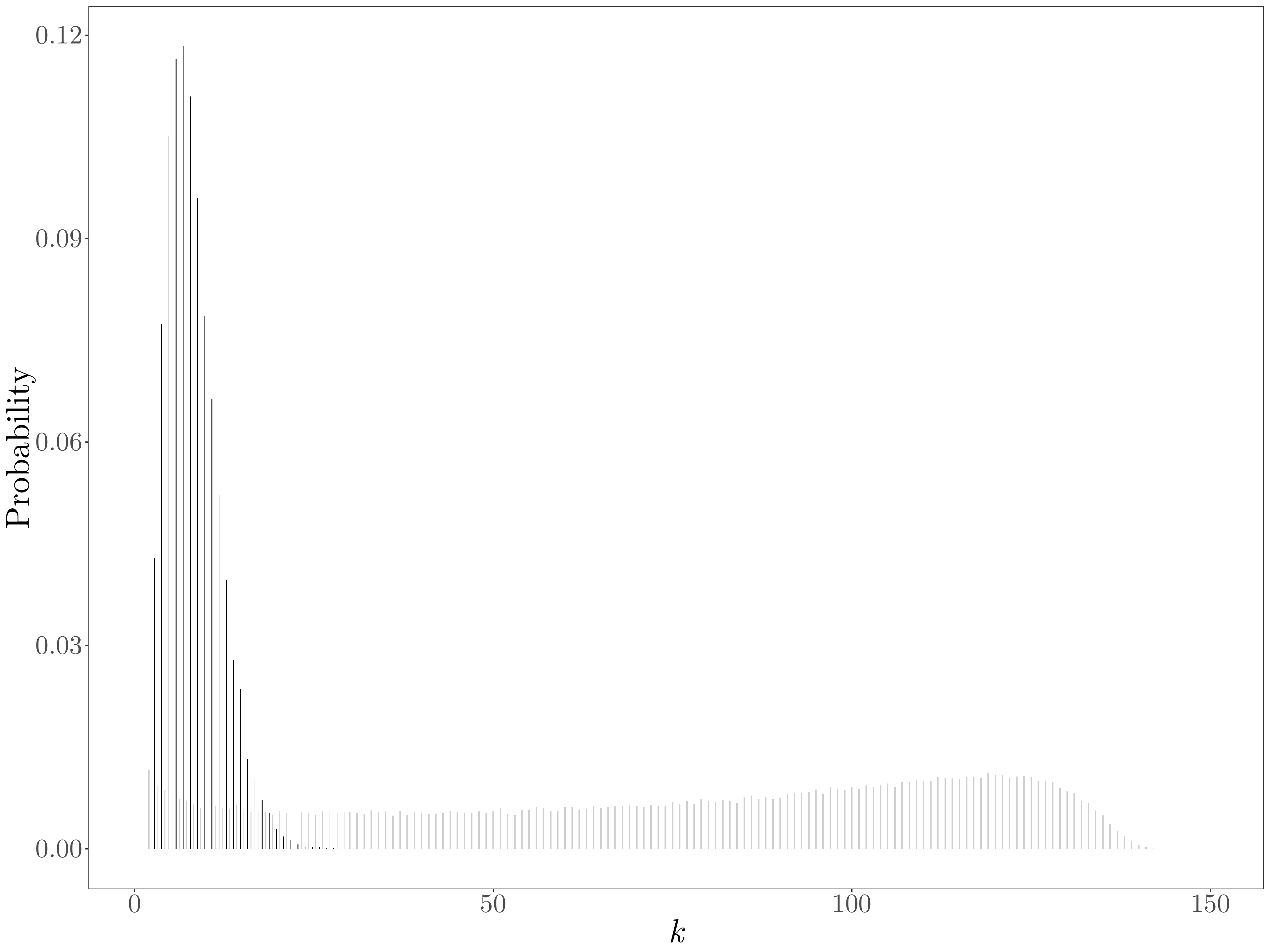}
\caption{Prior and posterior distributions for the number of clusters $K$, in gray and black respectively, for patrol area 58.}
\label{fig:postk58}
\end{figure}

With regard to the forecasting performance within this dataset, Tables \ref{tab:mae_1}, \ref{tab:mae_2}, and \ref{tab:mae_3} present the mean absolute deviations for one, two and three-steps-ahead forecasts for both the DP-INAR(1) and the original INAR(1) model. In these tables, the mean absolute deviations are computed predicting the values of the last $42$, $43$, and $44$ months of each time series, according to the desired number of steps ahead, using the cross-validation procedure described in Section \ref{sec:forecast}. The tables show that the DP-INAR(1) model outperforms the INAR(1) in the majority of the districts.

\section*{Acknowledgements}

Helton Graziadei and Hedibert F. Lopes thank FAPESP  for financial support through grants numbers 2017/10096-6 and 2017/22914-5. 

\clearpage

\bibliographystyle{ieeetr}

\begin{table}[t!]
\caption{Mean absolute deviations for {\bf one-step-ahead} predictions. The first block is formed by  police patrol areas in which the DP-INAR(1) model outperforms the INAR(1) model. The DP-INAR(1) model produces lower mean absolute deviations in $67\%$ of the patrol areas.}
\vspace{0.5cm} 
\centering 
\begin{tabular}{ccc}
\hline \hline
Patrol Area & DP-INAR(1)         & INAR(1)      \\
\hline \hline 
21       & \textbf{1.1395} & 1.1860          \\
27       & \textbf{1.1860} & 1.3488          \\
25       & \textbf{1.2326} & 1.3023          \\
26       & \textbf{1.5116} & 2.0233          \\
24       & \textbf{1.5349} & 1.6512          \\
44       & \textbf{1.7907} & 1.8372          \\
33       & \textbf{1.8140} & 1.9302          \\
56       & \textbf{1.9302} & 2.0930          \\
16       & \textbf{2.0000} & 2.0930          \\
22       & \textbf{2.1163} & 2.2791          \\
17       & \textbf{2.2558} & 2.2791          \\
47       & \textbf{2.2558} & 2.3023          \\
58       & \textbf{2.5116} & 2.9767          \\
14       & \textbf{2.5349} & 2.5814          \\
54       & \textbf{2.5349} & 2.8837          \\
46       & \textbf{2.6279} & 2.7442          \\
15       & \textbf{2.7209} & 2.7907          \\
23       & \textbf{3.2093} & 3.3023          \\
31       & \textbf{3.4419} & 3.4884          \\
12       & \textbf{3.5116} & 3.9070          \\
\hline \hline 
28       & 0.8372          & \textbf{0.8140} \\
43       & 2.1861          & \textbf{2.1628} \\
41       & 2.3953          & \textbf{2.3721} \\
13       & 2.6977          & \textbf{2.6744} \\
53       & 2.8837          & \textbf{2.8372} \\
51       & 2.9302          & \textbf{2.8605} \\
32       & 3.4884          & \textbf{3.4419} \\
34       & 3.6744          & \textbf{3.5814} \\
52       & 3.9302          & \textbf{3.8140} \\
55       & 4.8837          & \textbf{4.5116} \\ 
\hline \hline 
\label{tab:mae_1}
\end{tabular}
\end{table} 

\begin{table}[t!]
\caption{Mean absolute deviations for {\bf two-step-ahead} predictions. The first block is formed by police patrol areas in which the DP-INAR(1) model outperforms the INAR(1) model. The DP-INAR(1) model produces lower mean absolute deviations in $61\%$ of the areas.}
\vspace{0.5cm} 
\centering 
\begin{tabular}{ccc}
\hline \hline 
Patrol Area & INAR(1) & DP-INAR(1) \\
\hline \hline 
11       & \textbf{1.1429} & 1.1667          \\
21       & \textbf{1.1667} & 1.3571          \\
27       & \textbf{1.3095} & 1.3810          \\
24       & \textbf{1.7381} & 1.9524          \\
44       & \textbf{1.8333} & 1.8810          \\
33       & \textbf{1.9524} & 2.2381          \\
41       & \textbf{2.2619} & 2.3095          \\
26       & \textbf{2.2857} & 2.5714          \\
56       & \textbf{2.3095} & 2.5476          \\
22       & \textbf{2.3333} & 2.5714          \\
13       & \textbf{2.6190} & 2.8095          \\
15       & \textbf{2.6667} & 2.8095          \\
51       & \textbf{2.7857} & 2.8095          \\
54       & \textbf{2.9286} & 3.4524          \\
58       & \textbf{2.9762} & 3.5238          \\
32       & \textbf{3.7143} & 3.7619          \\
12       & \textbf{3.9286} & 4.4286          \\
\hline \hline 
25       & 1.2619          & \textbf{1.2381} \\
17       & 2.2143          & \textbf{2.1429} \\
47       & 2.2857          & \textbf{2.2619} \\
14       & 2.4048          & \textbf{2.3571} \\
46       & 2.7143          & \textbf{2.5952} \\
29       & 2.9762          & \textbf{2.9524} \\
42       & 3.3571          & \textbf{3.3333} \\
31       & 3.6905          & \textbf{3.6667} \\
34       & 4.0000          & \textbf{3.8810} \\
55       & 4.4762          & \textbf{4.0000} \\
52       & 4.1667          & \textbf{4.0714} \\
\hline \hline 
\label{tab:mae_2}
\end{tabular}
\end{table}

\begin{table}[t!]
\caption{Mean absolute deviations based on {\bf three-step-ahead} predictions. The first block is formed by police patrol areas in which the DP-INAR(1) model outperforms the INAR(1) model. The DP-INAR(1) model produces lower mean absolute deviations in $70\%$ of the areas.}
\vspace{0.5cm} 
\centering 
\begin{tabular}{ccc}
\hline \hline
Patrol Area & DP-INAR(1)      & INAR(1)       \\
\hline \hline 
11       & \textbf{1.1463} & 1.1951          \\
21       & \textbf{1.1707} & 1.4146          \\
25       & \textbf{1.2683} & 1.2927          \\
27       & \textbf{1.3659} & 1.4146          \\
24       & \textbf{1.8293} & 1.9512          \\
33       & \textbf{2.1220} & 2.3171          \\
17       & \textbf{2.1707} & 2.1951          \\
22       & \textbf{2.3171} & 2.5610          \\
56       & \textbf{2.3415} & 2.6341          \\
26       & \textbf{2.4390} & 2.7073          \\
13       & \textbf{2.6585} & 2.9024          \\
15       & \textbf{2.7561} & 2.9024          \\
53       & \textbf{2.9024} & 2.9756          \\
58       & \textbf{3.0000} & 3.6341          \\
54       & \textbf{3.0244} & 3.3902          \\
23       & \textbf{3.4146} & 3.5366          \\
31       & \textbf{3.5610} & 3.5854          \\
32       & \textbf{3.7805} & 3.9024          \\
12       & \textbf{4.0732} & 4.4634          \\
\hline \hline 
44       & 1.9268          & \textbf{1.9024} \\
43       & 2.3171          & \textbf{2.2683} \\
14       & 2.5610          & \textbf{2.4878} \\
45       & 2.6098          & \textbf{2.5854} \\
46       & 2.6098          & \textbf{2.5366} \\
52       & 4.0244          & \textbf{3.9268} \\
34       & 4.1707          & \textbf{4.1463} \\
55       & 4.2927          & \textbf{4.1951} \\
\hline \hline 
\label{tab:mae_3}
\end{tabular}
\end{table}

\end{document}